\documentclass[letterpaper,11pt]{article}
\usepackage{graphicx} %
\usepackage[]{amsmath,amssymb,amsfonts,latexsym,amsthm,enumerate,fullpage,xcolor,cite,bbm,mathtools}
\usepackage[pagebackref,colorlinks,citecolor=blue,urlcolor=magenta,linkcolor=magenta,,bookmarks=true]{hyperref}
\usepackage[nameinlink, noabbrev, capitalize]{cleveref}
\usepackage{amsmath}
\usepackage{booktabs}
\usepackage{pgfplots}
\title{Linear Hashing Is Poissonesque}
\author{Vinayak M. Kumar\thanks{%
Simons Institute for the Theory of Computing, UC Berkeley.
\href{mailto:vmkumar@berkeley.edu}{\texttt{vmkumar@berkeley.edu}}.
Supported by a Simons Institute Research Fellowship.
A preliminary version of this work appears in the author's thesis,
written at UT Austin with support from a Jane Street Fellowship and a UT Austin Dean's Prestigious Fellowship
Supplement.}}
  
\date{}
\renewcommand{\backref}[1]{}

\renewcommand{\backrefalt}[4]{%
\ifcase #1 %
\or
[p.\ #2]%
\else
[pp.\ #2]%
\fi}

\newtheorem{theorem}{Theorem}
\newtheorem{lemma}{Lemma}
\newtheorem{proposition}{Proposition}

\newtheorem{corollary}{Corollary}
\newtheorem{definition}{Definition}

\newtheorem{remark}{Remark}

\usepackage{comment}
\usepackage{nicefrac}
\crefname{prop}{Proposition}{Propositions}
\crefname{ineq}{inequality}{inequalities}
\creflabelformat{ineq}{#2(#1)#3}

\crefname{THM}{Theorem}{Theorems}

\newcommand{\E}{\mathbb{E}}

\newcommand{\Z}{\mathbb{Z}}
\newcommand{\F}{\mathbb{F}}

\newcommand{\R}{\mathbb{R}}

\newcommand{\bars}{\bar{S}}

\newcommand{\pois}{\text{Pois}}
\newcommand{\eps}{\varepsilon}
\newcommand{\Span}{\text{Span}}

\renewcommand{\subset}{\subseteq}

\newcommand{\cE}{\mathcal{E}}

\begin{document}

\maketitle

\begin{abstract}

    The Poisson approximation theorem states that when $m$ balls are tossed uniformly and randomly into $n$ bins, the distribution of the bin loads can sometimes be approximated by $n$ independent Poisson random variables of mean $m/n$. We show that when the balls are tossed via a random $\F_2$-linear map, the load distribution is still governed by this Poisson distribution, albeit in a (necessarily) weaker manner. This allows us to show that a random linear map behaves like a fully random map in the following ways:

\begin{itemize}
    \item \textbf{Expected loads.} We show that the expected value of the $r$th largest bin is governed by the $r$th largest value among $n$ independent Poisson random variables of mean $m/n$. Specializing to the canonical setting of $m=n$ and $r=1$, our bounds show that the expected maximum load when $n$ balls are hashed into $n$ bins by a random $\F_2$-linear map is no more than the corresponding expectation for a fully random map \emph{plus} $\frac{3}2 + o(1)$, drastically improving on bounds derived by the framework of Jaber, Zuckerman, and the author (STOC 2025). Specializing to $m\gg n\log n$, we positively resolve a question of Pathegama and Barg (TIT 2025) over $\F_2$, who had asked whether the expected max load is $(1+o(1))m/n$ in this regime.

    \item \textbf{Concentration.} When $m\gg n\log (1/\rho)$ balls are tossed into $n$ bins randomly, we expect with $0.99$ probability that \emph{every} density-$\rho$ subset of bins will contain a $(1\pm o(1) )\rho$-fraction of the balls. We show that this remains the case if the $m$ balls are tossed using a random linear map, and our bound matches the $o(1)$-factor produced by the fully random case. This gives a stronger version of the Leftover Hash Lemma for linear maps. An implication of this is that all bins have load $m/n \pm O(\sqrt{(m/n)\log n})$ with probability $0.99$, which strictly improves the two-sided bounds of Dhar and Dvir (FOCS 2022) over $\F_2$, and quantitatively matches the order of deviation for constant failure probability of fully random hashing.

\end{itemize}
To obtain these results, we first show that the average of a log-potential on the bin loads is upper bounded by the cumulant generating function of a Poisson random variable of mean $m/n$. Expectation and tail bounds then easily fall out of this inequality, similar to how Chernoff tail bounds fall out of moment generating function inequalities. 

\end{abstract}
\break
\section{Introduction}

How random is a random linear map? This fundamental question has been studied for more than seventy years, perhaps originating with the question of how random linear codes compare to fully random codes \cite{elias55}. Since then, this question has been studied extensively within coding theory \cite{zp81,ghsz, ghk,cgv,Woo13,rw14,liwoo,glmrsw,agglz,dlmnrr, silas, yuanzhu, km26}, as well as in other contexts such as linear-seeded extractors \cite{dewatson,wootters2014thesis,rao09b,li-affine,chenzuck,cgl21,CL22}
and linear hash functions \cite{cw79,mcw78,mv81,alon99,babka18,dhardvir,PaBa25,jkz25,bshouty26,lund, km26}.  Although a basic object, random linear maps have been proven to be tricky to analyze in all of these contexts. The main reason is because of the copious amounts of correlations they induce. Linear maps are not even 3-wise independent, since the output of vectors $x$ and $y$ dictate the output of $x+y$. In general, any subset of linearly dependent inputs induce correlations
among their outputs. In this paper, we focus on the hashing properties of a random $\F_2$-linear map.

 Consider tossing $m$ balls independently and uniformly into $n$ bins.  An extremely useful idea, dating back to the 1970s \cite{Holst77}, was to analyze this random process using Poisson random variables. After more decades of development, this culminated in the celebrated Poisson approximation theorem \cite[Section 5.4]{mitzupfal}. This theorem states that for tail events, the distribution of the loads of the $n$ bins behaves like that of $n$ independent Poisson random variables with mean $m/n$ (denoted $\pois(m/n)^{\otimes n}$). Therefore, many monotone events and occupancy statistics can be analyzed using independent Poisson variables, and the vast literature on the Poisson distribution can be used to analyze this random process.

It is then natural to ask whether an analogous tool can be developed if $m$ balls are tossed into $n$ bins using a random linear map instead. More formally, let $S$ be an arbitrary subset of size $m$ in some extremely large $\F_2$-vector space $U$, and let $V$ be an $\F_2$-vector space of size $n$. A random linear map $h: U\to V$ consequently ``tosses'' the $m$ balls, $S$, into the $n$ bins, $V$. For each $v\in V$ we define $L_v(S,h)\coloneqq |\{s\in S: h(s) = v\}|$ to be the load of bin $v$. For \emph{any} $S$ and random linear $h$, we are asking whether \[(L_v(S,h))_{v\in V}\stackrel{?}{\approx} \pois(m/n)^{\otimes n}\] for tail events. This would imply an exact analog of the Poisson approximation theorem for random linear maps. For linearly independent $S$, the problem reduces to fully random hashing; if $S$ consists of $m$ linearly independent vectors, $h$ tosses $S$ in an independent and uniform manner and the standard Poisson approximation theorem kicks in. Unfortunately, if we make no assumptions about $S$, the approximation is false. Consider $m=n$. By a Chernoff bound, one single $\pois(1)$ variable has sharply decaying tails: $\Pr[\pois(1)\ge t]\le (e/t)^t$ \cite[Theorem 5.4]{mitzupfal}. However, Bshouty showed that when $S$ is a subspace, the load of a fixed bin has exponentially worse tails: $\Pr_h[L_0(S,h)\ge t]\ge 2^{-O(\log^2 t)}$ \cite{bshouty26}.

    Thus, asking for the same tail bounds between these two distributions is too strong, but could it be the case that coarser statistics of these two distributions are in sync?
 
\subsection{A Potential-Poisson Inequality}

We provide an analytic notion in which the load distribution of linear hashing is governed by Poisson random variables. Consider the following potential, defined for any real $\lambda$, function $h: U\to V$ with $|V|=n$, and set $S\subset U$ with $|S| = m$. \[\phi_\lambda(S,h)\coloneqq \log\left(\frac{1}n \sum_{v\in V} e^{\lambda (L_v(S,h)-m/n)}\right) \] In other words, it is the log-average-exp of the centered bin loads. Similarly to how the moment and cumulant generating functions capture statistics of its corresponding distribution, the potential function $\phi_\lambda$ captures (coarser) statistics of the load distribution of $h$. The key inequality we show is that for all real $\lambda$ and all $S\subset U$ of size $m$, the expectation of this potential, over a random linear $h$, is upper bounded by the cumulant generating function of a centered Poisson random variable with mean $m/n$. \begin{equation}\E_h\left[\phi_\lambda(S,h)\right]\le  \frac{m}n (e^\lambda -\lambda - 1)\label{eqn:potential-intro}\end{equation} The proof of this inequality is surprisingly simple. Furthermore, this inequality is quite versatile and can demonstrate various Poissonesque behaviors of the bin loads using the following three-fold recipe:

\begin{enumerate}
    \item Find statistics of the bin loads that are governed by the potential.
    \item Use \cref{eqn:potential-intro} to derive expectation bounds or tail bounds on the aforementioned statistics in terms of $\lambda$.
    \item Optimize the setting of $\lambda$ to produce the tightest bounds.
\end{enumerate}

This is analogous to how Chernoff-type bounds are derived from the moment generating function. That framework is not applicable here because the moment generating function for the loads of random linear maps can be extremely large for certain choices of $S$ (a simple example is when $S$ is a subspace). Despite this, we show that the logarithmic potential $\phi_\lambda$ can be effectively bounded, and that this potential bound can be harnessed to provide tight expectation and tail bounds on various load statistics. 

To demonstrate the versatility of this method, we instantiate this recipe to yield a collection of results. With respect to the first step of the recipe, we show that the average of the heaviest/lightest $r$ bins is governed by the potential function. Applying the second and third steps then allows us to derive sharp expectation and tail bounds on these statistics.

\subsection{Results}
We prove a variety of results on the load distribution of linear hashing. Previously, these results had their respective family of techniques. Our method provides a simplification and unification of these directions. We also believe that the following are just a few examples of this technique, and there are more applications to be discovered. We also mention that the following results hold for \emph{any} hash family for which \cref{eqn:potential-intro} holds (see the end of \cref{sec:technique} for other examples of such hash families).
\subsubsection{Expected Loads}

 Let $L_r^+(S,h)$ be the average load of the $r$ heaviest bins (i.e., the average of the largest $r$ values among the multiset $\{(L_v(S,h)\}_{v\in V}$), and let $L_r^-(S,h)$ be the average load of the $r$ lightest bins. Among $n$ independent $\pois(m/n)$ variables, define $T_r^+(m,n)$ and $T_r^-(m,n)$ to be what the Chernoff upper and lower tail bound predicts the $r$th largest and smallest value is, respectively.\footnote{ See \cref{def:threshold} for the formal definition.} We show the following theorem.

    \begin{theorem}
    \label{thm:main-informal}
        Let $u\ge \ell\ge 1$, $m\ge 1$, $n\coloneqq 2^\ell$, and $1\le r\le n$ be integers, and let $S\subset\F_2^{u}$ be an \emph{arbitrary} set of cardinality $m$. For a uniformly random linear $h: \F_2^u\to\F_2^\ell$,  \[T_r^-(m,n)\le \E[L_r^-(S,h)]  \le \E[L_r^+(S,h)] \le T^+_r(m,n).\]
    \end{theorem}

In other words, the the average loads of the $r$ heaviest and $r$ lightest bins are governed by the corresponding quantities over $n$ independent $\pois(m/n)$ variables. The proof is a straightforward execution of our recipe.

To concretely display how strong these expectation bounds are, we specialize to the canonical setting of $m=n$ and $r=1$. Jaber, Kumar, and Zuckerman recently showed that when hashing $n$ balls into $n$ bins, random linear maps have expected maximum load
$O(\log n/\log\log n)$ \cite{jkz25}. This matches fully random hashing up to a constant factor in expectation and resolves a question of Alon, Dietzfelbinger, Miltersen,
Petrank, and Tardos \cite{jkz25,alon99}. Bshouty subsequently refined this approach and obtained an upper bound matching the dominant second-order term of fully random hashing \cite{bshouty26}.\footnote{The precise quantity achieved is $\frac{\log n}{\log\log n}\left(1 + (1+o(1))\frac{\log\log\log n}{\log\log n}\right)$.} 
A corollary of \cref{thm:main-informal} is that the expected max load when $n$ balls are tossed into $n$ bins via a random $\F_2$-linear map is at most the corresponding fully random expectation \emph{plus} 1.5.

\begin{theorem}
\label{thm:headline}
     Let $u\ge \ell\ge 1$ be integers, and $n\coloneqq 2^\ell$.  Let $h: \F_2^u\to\F_2^\ell$ be a uniformly random linear map, and let $f: [n]\to[n]$ be a fully random map. There exists $c_\ell\in [0,1)$ such that for every $S\subset\F_2^{u}$ of cardinality $n$, \[\E\left[\max_{y\in \F_2^\ell }|h^{-1}(y)\cap S|\right] \le \E\left[\max_{y\in [n] }|f^{-1}(y)|\right] + \frac{1}2 + c_\ell + o_\ell(1).\] 

Furthermore, there is an increasing sequence $(\ell_j)$ such that $c_{\ell_j}\to 0$.     
\end{theorem}

Hence, the expected max load of linear hashing exceeds that of fully random hashing by at most  $\frac{3}2 + o(1)$ for sufficiently large $n$, and by at most $\frac{1}2 + o(1)$ for infinitely many $n$. The fully random expectation term is unavoidable: when $u\ge n$ and $S$ consists of $n$ linearly independent vectors, $h|_S$ is a uniformly random map. It is possible that the additive term could be improved.

\cref{thm:main-informal} provides very tight max load bounds in any regime of $m$ balls and $n$ bins. For example, the constant-additive comparison holds when the number of balls is sufficiently sparse compared to the number of bins (see \cref{rem:sparse-regime}). In the dense regime of $m\gg  n \log n$ balls into $n$ bins, the expected max load for a fully random map is extremely close to $m/n$. Unfortunately, the previous expected max load upper bound in this regime was only
$16m/n$, suffering a large constant-factor loss \cite{jkz25}. This loss arises from the aforementioned
tail-bound approach: in this regime, available tail bounds are too weak to integrate and yield a sharp expectation. By bypassing tail bounds, our
argument places the expected maximum and minimum loads within second-order
correction terms of the mean. This resolves a question explicitly posed by Pathegama and Barg over $\F_2$ \cite[Remark 6]{PaBa25}.

\subsubsection{Concentration Bounds}

When $m$ balls are uniformly and independently tossed into $n$ bins, the dream scenario is if each bin ends up with exactly $m/n$ balls. When $m\ge n\log n$, a proxy for this dream happens with high probability using simple Chernoff bounds. 

\begin{table}[t]
      \centering
      \caption{
        Comparison of deviation bounds when hashing a set $S\subset\F_2^u$ of $m$ balls into $n$ bins via a random $\F_2$-linear map $h: \F_2^u\to\F_2^{\log_2 n}$. We denote $\mu\coloneqq m/n$ and $\Delta=\max_v |L_v(S,h)-\mu|$. The results marked $\dagger$ were obtained independently.
        }
    \setlength{\tabcolsep}{5pt}
    \renewcommand{\arraystretch}{1.3}
    \vspace{5mm}
    \begin{tabular*}{0.8\linewidth}{
        @{\extracolsep{\fill}}llll@{}
    }
        \toprule
        
        & Ball/bin regime
        & Bound on $\Delta$
        & Failure probability \\
        \midrule
        \multicolumn{4}{l}{\emph{Fully random maps}} \\
        \addlinespace[2pt]

        Chernoff bound
        & $\mu>\log n$
        & $\lesssim \sqrt{\mu\log n}$
        & $n^{-\Omega(1)}$ \\

        \midrule
        \multicolumn{4}{l}{\emph{Random $\F_2$-linear maps}} \\
        \addlinespace[2pt]

        \cite{dhardvir}
        & $\mu\gtrsim \frac{\log^4(m/\eps)}{\eps^8}$
        & $\le \frac{\mu}{\eps^{1/2}}$
        & $\eps$ \\

        \cite{jkz25}
        & $\mu\gtrsim \frac{\log n}{\eps^{74}}$
        & $\lesssim \frac{\mu}{\eps^{1/2}}$
        & $\eps$ \\

        \cite{dlmnrr}
        & $\mu\ge 2^{\Omega(\sqrt{u})}$
        & $\lesssim \mu^{11/12}$
        & $\lesssim \frac{1}{\mu^{1/3}} + \frac{1}n$ \\

        \addlinespace[3pt]
        \cite{lund}\textsuperscript{$\dagger$}
        & $\mu\gtrsim \frac{\log n}{\eps}$
        & $\lesssim \sqrt{\frac{\mu}{\eps}\log n}$
        & $\eps$ \\

        \textbf{This work}\textsuperscript{$\dagger$}
        & $\mu>\log n$
        & $\le 5\sqrt{\frac{\mu}{\eps}\log n}$
        & $\eps$ \\

        \bottomrule
 \end{tabular*}
    
    \label{tab:two-sided-load-bounds}
\end{table}
\begin{theorem}[Informal Chernoff Estimate]
\label{thm:random-twosided}
    Let $n\ge 1, m > n\log n$ be integers. Let $f:[m]\to [n]$ be a random function. We have \[\Pr_f\left[\forall y\in [n], \left||f^{-1}(y)| - \frac{m}n\right|\le O\left(\sqrt{\frac{m\log n}{n}}\right)\right]\ge 1-n^{-\Omega(1)}.\]
\end{theorem}
The regime $m > n\log n$ is when the deviation term starts to become smaller than the mean for a fully random map.

Can a guarantee like \cref{thm:random-twosided} be obtained with random $\F_2$-linear maps? This question was first studied in depth by Dhar and Dvir \cite{dhardvir}. They achieve the first such bounds by using probabilistic tools to reduce the problem to rank-2 Furstenberg set lower bounds, and then applying a tailored version of the polynomial method. Later works of Jaber, Kumar, and Zuckerman \cite{jkz25} and of Doron, Leonov, Mosheiff, Navas, Resch, and Ribeiro \cite{dlmnrr}, showed incomparable bounds using an exponential and quadratic potential function, respectively.

Our method provides the following bound with a simple proof.

\begin{theorem}
\label{thm:intro-twosided}
     Let $u\ge \ell\ge 1$, $n\coloneqq 2^\ell$, $m> n\log n$ be integers, and $S\subset\F_2^{u}$ be an arbitrary set of cardinality $m$. For every $\eps\in(0,1)$, a uniformly random linear map $h: \F_2^u\to\F_2^\ell$ satisfies \[\Pr\left[\forall v\in \F_2^\ell, \left|L_v(S,h) - \frac{m}n\right| \le 5\sqrt{\frac{m\log n}{\eps n}}\right]\ge 1-\eps.\]
\end{theorem}

For constant $\eps$, the deviation has the same order in $m$ and $n$ as for a fully random map. Our deviation bounds kick in once $m > n\log n$, just like as in the fully random case. The dependence on $\eps$ is worse here than for a fully random map: we achieve probability $0.99$ whereas \cref{thm:random-twosided} achieves probability $1-n^{-\Omega(1)}$. We prove that this $\eps^{-1/2}$ dependence is necessary for random linear maps (see \cref{thm:mixing-tight}). We note that independent work of Lund \cite{lund} achieves bounds of weaker, but similar shape. They follow a similar framework of Dhar and Dvir, but imports stronger Furstenberg set lower bounds established by Kumar and Mon \cite{km26}.  

 We actually show a more general phenomenon than the two-sided bounds above. Consider random linear $h: \F_2^u\to\F_2^\ell$ and arbitrary $A\subset \F_2^u$ of size $m$. If $h$ behaved like a truly random function, then for every subset $B\subset\F_2^\ell$ of density $\rho$, $h$ would map about a $\rho$-fraction of $A$ to $B$. A collision probability argument akin to that of the Leftover Hash Lemma \cite{BBR88,ill89} (e.g., a straightforward application of  \cite[Lemma 4.2]{cmv13-simplehashingonentropicsources}) provides the following concentration guarantee. 

\begin{theorem}
\label{thm:intro-density}
    Let $u\ge \ell \ge 1$, $n\coloneqq 2^\ell$ be integers, let $A\subset \F_2^u$ be of size $m$. Let $\rho\in [0,1/2]\cap \frac{1}n\Z$ and $\eps\in (0,1)$. For a uniformly random linear $h:\F_2^u\to\F_2^\ell$, we have \[\Pr_h\left[ \forall B\subset \F_2^\ell\text{ with }|B| = \rho n, \left|\Pr_{a\sim A}[h(a)\in B] - \rho\right| \le  \sqrt{\frac{ \rho n}{\eps m}}\right]\ge 1 - \eps.\]
\end{theorem}
    By passing from $B$ to its complement in $\F_2^\ell$, it suffices to consider $\rho\le 1/2$.
    
    For simplicity, take $\eps$ constant. If $\rho = \Omega(1)$, this is indeed the best deviation bound one can obtain. However, the bound degrades for $\rho = o(1)$, and when $\rho  = O(n/m)$, the deviation bound is larger than the average density itself. We show the following sharpening.

\begin{theorem}
\label{thm:mixing-intro}
   Let $u\ge \ell \ge 1$ and $m\ge n\coloneqq 2^\ell$ be integers. Let $\eps\in (0,1)$ and $\rho\in (0,1/2]\cap \frac{1}n\Z$. Let $A\subset \F_2^u$ be of size $m$, and let $h:\F_2^u\to\F_2^\ell$ be a random linear map. If $m\ge n\log(1/\rho)$, we have\[\Pr\left[\forall B\subset \F_2^\ell\text{ with }|B| =\rho n, \left|\Pr_{a\sim A}[h(a)\in B] - \rho\right| \le 5\rho\sqrt{\frac{n \log (1/\rho)}{\eps m}}\right]\ge 1 - \eps.\]
\end{theorem}

Similarly to \cref{thm:intro-twosided}, for $\eps$ constant, the deviation bound matches the guarantees of a fully random map up to constants. While the deviation bound for a fully random map has a much better dependence on $\eps$, the $\eps^{-1/2}$ dependence can be shown to be nearly optimal by taking $A$ to be linearly independent shifts of a subspace (see \cref{thm:mixing-tight}). 

We remark that plugging in $\rho = 1/n$ into \cref{thm:mixing-intro} yields \cref{thm:intro-twosided}.

\subsection{Proof Technique}
\label{sec:technique}
We focus on how our framework gives a simpler proof of optimal expected max load, and compare it to previous approaches.
 
    First consider random \emph{surjective} linear maps. We follow a structure similar to previous potential-based approaches to random linear codes
and linear hashing \cite{ghsz,liwoo,jkz25,bshouty26}, where we
inspect how a potential evolves as the hash function is
revealed ``one kernel vector at a time''. This potential serves as a proxy for the max load. We would thus like to show the final potential is likely small, which then implies the max load is likely small.
The hashing analyses of \cite{jkz25,bshouty26} use an exponential potential $\Phi_i$
satisfying the conditional squaring inequality
\[
    \E[\Phi_{i+1}\mid v_{\le i}]\le \Phi_i^2,
\]
where $v_{\le i}$ denotes the kernel vectors revealed so far.
This readily gives an existence argument: at each step, some
choice of the next kernel vector makes the potential at most
square. Hence, kernel vectors can be fixed such that $\Phi_i\le \Phi_0^{2^i}$. However, passing from this existence argument to a high probability bound requires additional work. It is tempting to compose the squaring recurrence to deduce an expectation bound of $\E[\Phi_i]\le \Phi_0^{2^i}$, but this ``freshman's dream'' is false in general. One can construct a sequence of three variables satisfying conditional squaring in which the third random variable has infinite expectation. The most technically involved
part of the arguments in
\cite{liwoo,jkz25,bshouty26} was to use conditional squaring, along with other strong monotonicity properties of these potentials, to establish high probability guarantees on the final potential.

Departing from these works, we instead use the logarithm of
the exponential potential. Setting $\phi_i\coloneqq\log\Phi_i$, conditional squaring and Jensen's inequality give
\[
    \E[\phi_{i+1}\mid v_{\le i}] \le  2\phi_i.
\]
By linearity of expectation, this linear recurrence can indeed be composed to bound $\E[\phi_i]\le 2^i \phi_0$. Hence, we obtain an expectation bound on the final potential almost immediately. Together with an estimate of the initial potential, this yields \cref{eqn:potential-intro}. The max load is bounded by a linear function of the log-potential, so the potential bound directly yields an expected max load bound. A conceptual contribution of ours is that \cref{eqn:potential-intro} can be used to obtain a broader collection of results than just expected max load.

    To remove the surjectivity assumption, we show that a broad range of load statistics on random surjective linear maps can be lifted \emph{losslessly} to unrestricted random linear maps (see \cref{app:removing-surjectivity}). This is a bit counterintuitive in light of the fact that a random linear map has some chance of mapping all the balls into a very small number of bins.

We note that the exponential potential used in previous work \cite{jkz25,bshouty26} provides stronger tail bounds when hashing $m = o(n\log n)$ balls into $n$ bins, versus using our logarithmic potentials. However, our log-potentials provide stronger tail bounds in the dense regime when $m\ge n\log n$. We also note that recent work \cite{dlmnrr} uses a \emph{quadratic} potential to show near-optimal concentration bounds in the regime $m \ge n^{1+\eps}$ and over $\F_p$ for any prime $p$. Our approach gives optimal concentration bounds for $m$ as small as $n\log n$, but only works over $\F_2$. 

   The proof of \cref{eqn:potential-intro} extends to hash families beyond random linear maps. Fix integers $k\ge 2$ and  $u\ge \ell\ge 1$, and consider the following hashing procedure mapping $[k^u]$ to $[k^\ell]$. First, start with $k^u$ bins, where each of the $k^u$ universe elements are in their own bin. For $u-\ell$ iterations, \emph{pseudorandomly} partition the bins into groups of $k$, and merge each group into one bin. A \emph{pseudorandom} partition is one where the $k-1$ groupmates of any fixed bin form a uniformly random $(k-1)$-subset of the other current bins (no independence between distinct groups are required). Each iteration reduces the number of bins by a factor of $k$, leaving $k^\ell$ bins at the end. For $k=2$, this framework includes the load distribution of uniformly random surjective $\F_2$-linear maps \cite[Section 1.3]{jkz25}. The same potential argument proves \cref{eqn:potential-intro} for any hash family generated by this procedure with $n\coloneqq k^\ell$, yielding all expectation and concentration bounds derived from that inequality. Unfortunately, this hash family class for $k > 2$ does not capture $\F_k$-linear maps.
\section{Preliminaries}

  All logarithms are base $e$ unless otherwise specified. For integer $n$ we denote $[n] = \{1,\dots, n\}$, and for real $\alpha$ we denote $\{\alpha\} = \alpha - \lfloor \alpha\rfloor$. For a list of vectors $v_1,\dots, v_k$, we will denote the tuple $v_{\le i} \coloneqq (v_1, v_2, \dots, v_i)$ and $V_i = \Span(v_{\le i})$ (by convention, $V_0 = \{0\}$). $\F_2$ is the finite field of two elements. For a set $S$, we write $s\sim S$ to denote that $s$ is sampled uniformly at random from $S$. We will use the inequality $1+x\le e^x$ for all $x$. For a set $S$, we will denote \[\Sigma^S(m) = \left\{w\in \Z_{\ge 0}^S: \sum_{s\in S} w_s = m\right\}.\]
 
\paragraph{Load Statistics} For $f: A\to B$ and $S\subset A$, define $L_b(S,f)\coloneqq |f^{-1}(b)\cap S|$ to be the load of bin $b$. For $1\le r\le |B|$, denote by $L^{(r)}(S,f)$ the $r$th element of the multiset $\{L_b(S,f)\}_{b\in B}$ when sorted in increasing order. Thus, $L^{(r)}(S,f)$ is the number of balls in the $r$th lightest bin when $f$ is used to hash $S$. We now define \[L^+_r(S,f) \coloneqq \frac{1}r\sum_{n-r < i \le n} L^{(i)}(S,f)\] and \[L^-_r(S,f) \coloneqq \frac{1}r\sum_{1\le i\le r} L^{(i)}(S,f)\] to be the average of the $r$ heaviest and $r$ lightest bins, respectively. These will be key quantities used later.
 
 \paragraph{Poisson Random Variables} We will denote $\pois(\mu)$ to be a Poisson random variable with mean $\mu$. We will denote $D(\nu,\mu)\coloneqq \nu\log(\nu/\mu) - \nu+\mu$ to be the Kullback-Leibler divergence between $\pois(\nu)$ and $\pois(\mu)$ (see, e.g., \cite{Kullback1968}), with the convention $0\log 0\coloneqq 0$. We will use the following basic estimate. 
 \begin{proposition}
\label{prop:kl-est}
    For $\mu > 0$ and $t > s>0$, \[D(t,\mu)-D(s,\mu) = \int_s^t \log(x/\mu) dx\ge  (t-s)\log(s/\mu).\]
\end{proposition}

\begin{proof}
    Using the integral of $\log x$, we see \[D(t,\mu)-D(s,\mu) = t\log(t/\mu) - t - s\log(s/\mu) + s = \int_s^t \log(x/\mu) dx\ge (t-s)\log(s/\mu).\]
\end{proof}

\subsection{Removing Surjectivity For Free}

Intuitively, one may expect that surjectivity can be assumed in exchange for a small loss in parameters, as random linear maps are surjective with decent probability. However, to prove our extremely tight bounds, we cannot withstand even the smallest parameter degradation. We prove that bounds on functions of the loads of random surjective linear maps can be lifted to general linear maps with \emph{absolutely no} loss of parameters. We believe this lemma to be interesting in its own right.

    \begin{lemma}
\label{app:removing-surjectivity}
    Let $q$ be a prime power, and $\ell,m\ge 1$ and $1\le r\le q^\ell$ be integers. For any function $g: \F_q^u\to\F_q^\ell$ and $S\subset \F_q^u$ of size $m$, denote \[\vec{L}(S,g) \coloneqq (L_v(S,g))_{v\in \F_q^\ell}\in \Sigma^{\F_q^\ell}(m).\] Let $f: \Sigma^{\F_q^\ell}(m)\to[0,\infty)$ be an arbitrary function, and assume there exists $a\in [0,\infty)$ such that for infinitely many domain dimensions $U\ge \ell$, every $S\subset \F_q^U$ of size $m$, and uniformly random surjective linear map $h: \F_q^U\to\F_q^\ell$, we have \[\E\left[f\left(\vec{L}(S,h)\right)\right]\le a.\] Then for \emph{every} $u\ge \ell$, $S\subset\F_q^u$ of size $m$, and a uniformly random linear map $h:\F_q^u\to\F_q^\ell$, \[\E\left[f\left(\vec{L}(S,h)\right)\right]\le a.\]
    
\end{lemma}

 Because $f$ may be an arbitrary nonnegative function of the entire
load vector, \cref{app:removing-surjectivity} applies to a wide range
of statistics. For example, it transfers bounds on expected order statistics, polynomial or exponential moments, and tail probabilities, the latter by
taking $f$ to be the indicator of the relevant event. Lower bounds for
a bounded real-valued statistic $g$ can likewise be transferred by
applying the lemma to $M-g$, where
\[ M= \max_{\vec{w}\in\Sigma^{\F_q^\ell}(m)} g(\vec{w}).\]
In this paper, we will use \cref{app:removing-surjectivity} to transfer upper bounds on a carefully designed potential of the load vector.

 We note that weaker variants of this lemma have been established in \cite{jkz25, bshouty26, km26}, but incurred losses in parameters. Our proof contains ingredients from them, along with a new ``limit trick.'' 

 \begin{proof}[Proof of \cref{app:removing-surjectivity}]
Consider arbitrary $u\ge \ell$ and set $S\subset \F_q^u$ of size $m$. There are infinitely many $U\ge u$ satisfying the lemma's assumption (call such $U$ \emph{good}). Consider some good $U$ and take an arbitrary $u$-dimensional subspace $V\leq \F_q^U$. There is an isomorphism $\iota: \F_q^{u}\to V$. Denote $T\coloneqq \iota(S)$. Since $T\subset V$, we have for any $H:\F_q^U\to\F_q^\ell$ and $y\in \F_q^\ell$ that \begin{align*}|(H|_V)^{-1}(y)\cap T|  =|H^{-1}(y)\cap V \cap T|  = |H^{-1}(y)\cap T|. \end{align*} Therefore, $\vec{L}(T, H|_V) = \vec{L}(T,H) $. For uniformly random $H$, $H|_V:V\to \F_q^\ell$ will be a uniform random linear map. Consequently, \begin{align*}\E_h\left[f\left(\vec{L}(S,h)\right)\right] = \E_H\left[f\left(\vec{L}(T, H|_V)\right)\right] = \E_H\left[f\left(\vec{L}(T, H)\right)\right]. \end{align*}

Let $\cE$ be the event that $H$ is surjective. When viewed as a $\ell\times U$ random $\F_q$-matrix, this is equal to the event all rows are linearly independent. Hence \[\Pr[\cE] = (1-q^{-U})(1-q^{-U+1})\dots (1-q^{-U+\ell-1})\ge 1 - \sum_{i=0}^{\ell-1} q^{-U+i} \ge 1-q^{-U+\ell}.\]

Let $ M\coloneqq \max_{\vec{w}\in \Sigma^{\F_q^\ell}(m)} f(\vec{w})$. We use the fact that $a$ bounds the surjective expectation to deduce
 \[\E_h\left[f\left(\vec{L}(S,h)\right)\right] = \E_H\left[f\left(\vec{L}(T,H)\right)\right] \le \Pr[\neg\cE]\cdot M + \Pr[\cE]\cdot a\le q^{-U+\ell}M + a. \]Recalling that there are infinitely many good $U\ge u$ satisfying the lemma's assumption, we can take $U\to \infty$ through an increasing sequence of good integers, yielding the desired statement.
\end{proof}

\section{The Potential Function Bound}
\label{sec:potential-bound}

For $h: \F_2^u\to\F_2^\ell$, $S\subset \F_2^u$ of size $m$, and $\mu\coloneqq m/2^\ell$, consider the following potential function, which is the logarithm of a centered version of the potential function used in prior works \cite{ghsz,liwoo,jkz25,bshouty26}.

\[\phi_\lambda(S,h)\coloneqq \log\left(\E_{y\sim \F_2^\ell} e^{\lambda(L_y(S,h) - \mu)}\right).\]

Taking logarithms is the key difference from prior works that allows us to obtain quantitatively tighter results with a simpler argument. Centering the loads gives us tighter control on concentration bounds in the regime where the number of balls is much larger than the number of bins.

We now prove that the average value of this potential is upper bounded by the cumulant generating function of a centered Poisson random variable. This is the main engine behind all our results.

\begin{theorem}
\label{thm:main-technical}
    Let $u\ge \ell\ge 1$, $m\ge 1$, $n\coloneqq 2^\ell$ be integers and denote $\mu \coloneqq m/n$. Let $h: \F_2^u\to\F_2^\ell$ be a uniformly random linear map. For any $\lambda\in \R$ and subset $S\subset\F_2^{u}$ of cardinality $m$,   \[\E_h\left[\phi_\lambda(S,h)\right]\le \mu(e^\lambda - \lambda - 1).\]
\end{theorem}

In order to prove this inequality, we will consider sampling a random surjective linear map ``one kernel vector at a time'', and track the potential. Consider $k$ linearly independent vectors $v_{\le k}\in (\F_2^u)^k$ and a subset $S\subset \F_2^u$. Recall $V_i = \Span(v_{\le i})$. For each $0\le i\le k$ define \begin{align*}\bars_i(x) & \coloneqq |(x+V_i)\cap S| - \frac{|S|}{2^{u-i}}\end{align*} Note that $\bars_0$ is the mean-centered indicator of $S$, and when $\F_2^u$ is hashed according to cosets of $V_i$,  $\bars_i(x)$ is the number of balls in $x$'s bin, centered by the average load.

For arbitrary $\lambda\in \R$, define the sequence of exponential potential functions \begin{equation*}\Phi_i = \Phi_i(S;\lambda; v_{\le i})\coloneqq \E_{x\sim \F_2^u}[e^{\lambda \bars_i(x)}].\end{equation*} We now show that $(\Phi_i)_{i\ge 0}$ is controlled in conditional expectation. This was shown in prior work for the uncentered potential \cite{jkz25}, and the proof is essentially the same.

\begin{lemma}[{\hspace{-.1mm}\cite[Lemma 3 modified]{jkz25}}]
\label{clm:expsquare}
    Let $v_{\le i}\in (\F_2^u)^i$ be a tuple of linearly independent vectors, and let $v_{i+1}\sim \F_2^u\setminus V_i$. We have \[\E_{v_{i+1} }[\Phi_{i+1} ] \le \Phi_i^2.\]
\end{lemma}

\begin{proof}

   Fix an arbitrary choice of $\lambda$. Denote the following quantities: \[E_i(x)\coloneqq e^{\lambda \bars_i(x)}, \qquad \Psi_{i+1}({v_{i+1}}) \coloneqq \E_x[E_i(x)E_i(x+{v_{i+1}})].\] Note that $\Phi_i = \E_{x\sim \F_2^u}[E_i(x)]$. For any $v \notin V_i$, $V_i$ and $V_i+v$ partition $V_{i+1}$ and so $\bars_{i+1}(x) = \bars_i(x)+\bars_i(x+v)$. Hence, for any such $v$, $E_{i+1}(x) = E_i(x)E_i(x+v)$. Thus, for ${v_{i+1}}\sim \F_2^u\setminus V_i$, \[\E_{v_{i+1}}[\Phi_{i+1}] = \E_{x,{v_{i+1}}}[E_{i+1}(x)] = \E_{x,{v_{i+1}}}[E_i(x)E_i(x+{v_{i+1}})] = \E_{v_{i+1}}[\Psi_{i+1}({v_{i+1}})]. \] Hence it suffices to show that \[\E_{v_{i+1} }[\Psi_{i+1}(v_{i+1})]  \le \Phi_i^2.\] To do so, we will first show the average of $\Psi_{i+1}$ over all of $\F_2^u$ is $\Phi_i^2$, and then show that $\Psi_{i+1}(v)\ge \Phi_i^2$ for all $v\in V_i$. By an averaging argument, this implies the average of $\Psi_{i+1}$ outside $V_i$ is at most $\Phi_i^2$ as desired.

   To prove the first claim, note that for uniform and independent $(x,v)\in (\F_2^u)^2$, so is $(x,x+v)$.  Hence, \begin{align*}\E_{v} [\Psi_{i+1}(v)]  = \E_{x,v}[E_i(x)E_i(x+{v})]  =\E_{x,{v}}[E_i(x)E_i({v})]   = \E_x[E_i(x)]\E_{v}[E_i({v})]  = \Phi_i^2 \end{align*} 
 We now show the second claim. Since $\bars_i$ is constant over cosets of $V_i$, it follows $E_i(x) = E_i(x+v)$ for any $v\in V_i$. Therefore for all $v\in V_i$, \[ \Psi_{i+1}(v) = \E_x[E_i(x)^2]  \ge (\E_x[ E_i(x)])^2 = \Phi_i^2,\]
   where the last inequality follows from Jensen's inequality. Both claims are now proven and we conclude.

\end{proof}

We are now ready to prove the main inequality.

\begin{proof}[Proof of \cref{thm:main-technical}]
     We will first assume $h$ is a uniformly random \emph{surjective} linear map, and then use \cref{app:removing-surjectivity} to remove the surjection assumption.
     Denote $k=u-\ell$. Using the vectors $v_{\le k}$ and set $S$, define $\{\bars_i(x)\coloneqq |(x+V_i)\cap S| - m/2^{u-i}\}_{0\le i\le k}$, the potentials $\{\Phi_i \coloneqq \E_{x\sim \F_2^u}[e^{\lambda \bars_i(x)}]\}_{0\le i\le k}$, and the log-potentials $\{\phi_i\coloneqq \log \Phi_i\}_{0\le i\le k}$. We note $\phi_k = \phi_\lambda(S,h)$. Therefore, proving the theorem statement is equivalent to proving \[\E[\phi_k]\le \mu(e^\lambda - \lambda - 1).\] 
    
      By Jensen's inequality and \cref{clm:expsquare}, we deduce for each $0\le i < k$, \[\E_{v_{i+1}}[\phi_{i+1}| v_{\le i}] = \E_{v_{i+1}}[\log(\Phi_{i+1})| v_{\le i}]\le \log(\E_{v_{i+1}}[\Phi_{i+1}|v_{\le i}])\le \log(\Phi_i^2) = 2\phi_i.\] Therefore, we can repeatedly apply the linearity of expectation and law of total expectation to deduce \begin{align}\E[\phi_k] = \E_{v_{\le k-1}}[\E_{v_k}[\phi_k|v_{\le  k-1}]] & \le 2\E_{v_{\le k-1}}[\phi_{k-1}] \nonumber \\ & \le 4\E_{v_{\le k-2}}[\phi_{k-2}] \nonumber \\ & \le \cdots \nonumber \\ & \le 2^k \phi_0 .\label{eqn:phi-exp-upper}\end{align}  
      We now focus on bounding $\phi_0$. As $|S| = m$, we can compute \[\Phi_0 = \E_x[e^{\lambda \bars_0(x)}] = \frac{m}{2^u}\cdot e^{\lambda(1-m/2^u)}+ \left(1-\frac{m}{2^u}\right)\cdot e^{\lambda(0- m/2^u)}= e^{-\lambda m/2^u}\left(1 + \frac{m(e^\lambda -1)}{2^u}\right) .\] Hence, by taking logarithms, we deduce \[\phi_0 =  \log\left(1 + \frac{m(e^\lambda -1)}{2^u}\right) - \frac{\lambda m}{2^u}\le\frac{m(e^\lambda -1)}{2^u}  - \frac{\lambda m}{2^u} = \frac{\mu(e^\lambda -\lambda - 1)}{2^k}.\label{eqn:phikbound}\] Combining with \cref{eqn:phi-exp-upper} yields the bound \begin{equation*}\E[\phi_k]\le \mu(e^\lambda - \lambda - 1).\end{equation*} Thus, the inequality for random \emph{surjective} linear maps is proved.

      For $\vec{a}\in \Sigma^{\F_2^\ell}(m)$, define $f(\vec{a}) = \log\left(\E_{y\sim \F_2^\ell} e^{\lambda(a_y - \mu)}\right)$. Note that by Jensen's inequality, \[f(\vec{a}) = \log\left(\E_{y\sim \F_2^\ell} e^{\lambda(a_y - \mu)}\right)\ge \E_{y\sim \F_2^\ell}\log(e^{\lambda (a_y-\mu)}) = \lambda(\E_{y\sim\F_2^\ell}[a_y] - \mu) = 0.\] Hence, $f$ maps $\Sigma^{\F_2^\ell}(m)\to [0,\infty)$. Applying \cref{app:removing-surjectivity} to $f$ and noting $\phi_\lambda(S,h) = f(\vec{L}(S,h))$, we deduce that the inequality also follows for random linear $h$ as well.
\end{proof}

\section{Following the Recipe}
The main inequality is proved. We recall the three-step plan to establishing linear hashing bounds: \begin{enumerate}
    \item Find load statistics that can be extracted from $\phi$.
    \item Use expectation and tail bounds on $\phi$ to obtain bounds on such properties.
    \item Optimize the choice of $\lambda$.
\end{enumerate} 
Stitched together, these steps will obtain sharp information regarding the load distribution when hashing with a random $\F_2$-linear map.

\subsection{Extracting Statistics out of $\phi$}

We show that the top-$r$ and bottom-$r$ load averages can be bounded by a linear function of $\phi$.
\begin{lemma}

\label{clm:load-potential}
    Let $u\ge \ell \ge 1,n\coloneqq 2^{\ell}, 1\le r\le n$ be integers, and let $f: \F_2^u\to \F_2^{\ell}$ be a function. For any $\star\in \{+,-\}$, we have
        \[\phi_\lambda(S,f)\ge \lambda(L^*_r(S,f) - \mu) - \log(n/r).\]

\end{lemma}

\begin{proof}
    We can lower bound \[\E_{y\sim \F_2^\ell}[e^{\lambda (L_y(S,f) - \mu)}] \ge \frac{1}{2^\ell}\sum_{n-r < i\le n} e^{\lambda (L^{(i)}(S,f) - \mu)} \ge \frac{r}{2^\ell}e^{\lambda (L_r^+(S,f) - \mu)}\] and \[\E_{y\sim \F_2^\ell}[e^{\lambda (L_y(S,f)-\mu)}] \ge \frac{1}{2^\ell}\sum_{1 \le i\le r} e^{\lambda (L^{(i)}(S,f) - \mu)} \ge \frac{r}{2^\ell}e^{\lambda (L_r^-(S,f) - \mu)},\] where the last step in each display is Jensen's inequality on the convex function $y\mapsto e^{\lambda y}$. Taking logarithms yields the desired result.
\end{proof}

\subsection{Bounds on Load Statistics}

We will now combine \cref{thm:main-technical} and \cref{clm:load-potential} to obtain expectation and tail bounds on the top-$r$ and bottom-$r$ load averages.

\begin{corollary}
\label{cor:exp}
     Let $u\ge \ell\ge 1$, $m\ge 1$, $n\coloneqq 2^\ell$, $1\le r\le n$ be integers and denote $\mu \coloneqq m/n$. Let $h: \F_2^u\to\F_2^\ell$ be a uniformly random linear map. For any subset $S\subset\F_2^{u}$ of cardinality $m$ and $\star\in \{+,-\}$, we have \[\sup_{\lambda < 0}\frac{\mu(e^\lambda - 1) + \log(n/r)}{\lambda} \le \E_h[L_r^\star(S,h)]\le \inf_{\lambda > 0}\frac{\mu(e^\lambda - 1) + \log(n/r)}{\lambda}. \]
\end{corollary}

\begin{proof}
    By \cref{clm:load-potential}, we have \[\phi_\lambda(S,h) \ge \lambda (L_r^\star(S,h) - \mu) - \log(n/r).\] Therefore, by \cref{thm:main-technical} and linearity of expectation, we have \[\mu(e^\lambda - \lambda - 1)\ge \E[\phi_\lambda(S,h)] \ge \lambda(\E[L_r^\star(S,h)] - \mu) - \log(n/r).\] For any $\lambda > 0$, it follows by isolating $\E[L_r^\star(S,h)]$ that \[\E[L_r^\star (S,h)]\le \frac{\mu(e^\lambda - 1) + \log(n/r)}{\lambda},\] and for any $\lambda < 0$, it follows  \[\E[L_r^\star (S,h)]\ge \frac{\mu(e^\lambda - 1) + \log(n/r)}{\lambda}.\] The corollary follows.
\end{proof}

We now derive upper and lower tail bounds.

\begin{corollary}
\label{cor:tail}
     Let $u\ge \ell\ge 1$, $m\ge 1$, $n\coloneqq 2^\ell$, and denote $\mu \coloneqq m/n$. Let $h: \F_2^u\to\F_2^\ell$ be a uniformly random linear map, and $S\subset\F_2^{u}$ be a subset of cardinality $m$. For any $\eps\in(0,1)$, and $\lambda > 0$, \[\Pr_h\left[\forall r\in[n]: ~L_r^+(S,h) < \mu  - \frac{\mu}\eps + \frac{\frac{\mu}\eps(e^\lambda - 1 ) + \log(n/r)}{\lambda}\right] \ge 1- \eps\] and for any $\lambda < 0$ we have \[\Pr_h\left[\forall r\in[n]: ~L_r^-(S,h) > \mu - \frac{\mu}\eps + \frac{\frac{\mu}\eps(e^\lambda - 1) + \log(n/r)}{\lambda}\right] \ge 1-\eps.\]  
\end{corollary}

\begin{proof}
    
    We will apply Markov's inequality on $\phi_\lambda(S,h)\coloneqq \log\left(\E_{y\sim \F_2^\ell} e^{\lambda (L_y(S,h) - \mu)}\right)$. We first verify $\phi_\lambda$ is nonnegative. Indeed, by Jensen's inequality, we have \[\log\left(\E_{y\sim \F_2^\ell} e^{\lambda (L_y(S,h) - \mu)}\right)\ge \E_{y\sim \F_2^\ell}\left[\log(e^{\lambda (L_y(S,h) - \mu)})\right] = \lambda \E_{y\sim \F_2^\ell}[L_y(S,h) - \mu] = 0.\] Hence, by Markov's inequality and \cref{thm:main-technical}, we have that for any $\eps\in (0,1)$, \[\Pr\left[\phi_\lambda(S,h) < \frac{\mu(e^\lambda - \lambda - 1)}\eps\right] \ge 1-\eps.\] Condition on this event. By \cref{clm:load-potential}, it follows that for \emph{all} $r\in[n]$, \[ \lambda (L_r^\star(S,h) - \mu) - \log(n/r) \le \phi_\lambda(S,h) < \frac{\mu(e^\lambda - \lambda - 1)}\eps.\] For any $\lambda > 0$, this simplifies to \[L_r^\star(S,h) < \mu + \frac{\frac{\mu}\eps(e^\lambda - \lambda - 1) + \log(n/r)}{\lambda} = \mu - \frac{\mu}\eps + \frac{\frac{\mu}\eps(e^\lambda - 1) + \log(n/r)}{\lambda} ,\] and for any $\lambda < 0$, the flipped inequality holds. 
\end{proof}

\subsection{Optimizing $\lambda$}
\label{sec:poisson}
We have now established expectation and tail bounds in terms of the free variable $\lambda$. It now remains to set $\lambda$ optimally. For the proofs, we will plug in the optimal values of $\lambda$ and have things ``magically'' work out. Afterwards, we will explain how the value is derived using the Cram\'{e}r transform.

In order to express these optimal values, we need to formally define some Poisson threshold quantities.  Recall $D(\nu,\mu)\coloneqq \nu\log(\nu/\mu) - \nu+\mu$ is the Kullback-Leibler divergence between $\pois(\nu)$ and $\pois(\mu)$.

\begin{definition}[Poisson thresholds]
\label{def:threshold}
    Let $m>0$ be real, let $n\ge 1$ and $1\le r\le n$ be integers, and put $\mu\coloneqq m/n$. We define $T^+_r(m,n)$ to be the unique real $t\in [\mu,\infty)$ such that \begin{equation}
\label{eqn:log}
     D(t,\mu) = \log(n/r).
\end{equation} We define $T^-_r(m,n)$ to be $0$ if $\mu\le \log (n/r)$, and otherwise to be the unique real $t\in (0,\mu]$ such that \cref{eqn:log} holds.
\end{definition}

The threshold $T^+_r$ is well defined: at $t=\mu$ the left-hand side of \cref{eqn:log} equals $0$, and it increases continuously to infinity on $[\mu,\infty)$. 
The lower threshold $T^-_r$ is also well defined: at $t=\mu$, the left-hand side of \cref{eqn:log} is $0$, and as $t$ approaches $0$, the left-hand side increases continuously to $D(0,\mu) = \mu$. Hence, when $\mu > \log(n/r)$ there is a unique solution. 

Intuitively, the quantity $T^+_r(m,n)$ is the load level at which the standard
Poisson Chernoff bound predicts $r$ heavy bins when hashing $m$ balls randomly into $n$ bins. Each bin's load can be approximated by a Poisson random variable with mean $\mu = m/n$. The upper tail Poisson Chernoff bound states that for $t\ge \mu$, $\Pr[\pois(\mu)\ge t]\le e^{-D(t,\mu)}$ \cite[Theorem 5.4]{mitzupfal}. Thus, the Chernoff bound predicts $r$ of the $n$ bins will have load $\ge t$ when $e^{-D(t,\mu)} = r/n$, which is equivalent to $t = T^+_r(m,n)$. Analogously, since the lower tail Poisson Chernoff bound states that for $t\le \mu$ $\Pr[\pois(\mu)\le t]\le e^{-D(t,\mu)}$, it follows that $T^-_r(m,n)$ is the load level at which the Poisson Chernoff bound predicts $r$ light bins.

 In summary, $T^+_r(m,n)$ and $T^-_r(m,n)$ are the expectation threshold estimates for the $r$th largest and $r$th smallest load in the fully random hash setting, respectively. We show that the $r$th largest and $r$th smallest loads of linear hashing are bounded by these expectation thresholds as well. In fact, we actually show something stronger: the top-$r$ and bottom-$r$ load averages are controlled by these Poisson thresholds.

\begin{theorem}[Poisson Approximation in Expectation]
\label{thm:exp-opt}
    Let $u\ge \ell\ge 1$, $m\ge 1$, $n\coloneqq 2^\ell$, and $1\le r\le n$ be integers, and let $S\subset\F_2^{u}$ be an arbitrary set of cardinality $m$. For a uniformly random linear $h: \F_2^u\to\F_2^\ell$,  \[T_r^-(m,n)\le \E_h[L_r^-(S,h)]\le \E_h[L_r^+(S,h)]\le T^+_r(m,n).\]
\end{theorem}

\begin{proof}
    Denote $\mu\coloneqq m/n$. It suffices to show $T_r^-(m,n)\le \E_h[L_r^\star(S,h)]\le T_r^+(m,n)$ for all $\star\in \{+,-\}$.  If $r=n$, then $T_r^\star(m,n) = \mu= L_r^\star(S,h)$ for every $h$, and the result follows. Hence, we can assume $r < n$, and therefore $T_r^-(m,n) < \mu < T^+(m,n)$. Fix arbitrary such $r$, and for $\star\in \{+,-\}$, denote $t^\star \coloneqq T_r^\star(m,n)$. 
    
    By setting $\lambda = \log(t^+/\mu) > 0$, we observe that by \cref{cor:exp} that  \[  \E_h[L_r^\star(S,h)]\le \inf_{\lambda > 0}\frac{\mu(e^\lambda - 1) + \log(n/r)}{\lambda}\le \frac{\mu(t^+/\mu - 1) + \log(n/r)}{\log(t^+/\mu)} = t^+, \] where the last equality follows from the fact $D(t^+,\mu) = \log(n/r)$. Now if $m \le n\log(n/r)$, then we trivially have $0 = T_r^-(m,n)\le  \E_h[L_r^\star(S,h)]$. Otherwise, $m > n\log(n/r)$, and so $t^- > 0$. Upon setting $\lambda = \log(t^-/\mu) < 0$, we see \[ \E_h[L_r^\star(S,h)] \ge \sup_{\lambda < 0}\frac{\mu(e^\lambda - 1) + \log(n/r)}{\lambda}\ge \frac{\mu(t^-/\mu - 1) + \log(n/r)}{\log(t^-/\mu)} = t^-. \] These computations, along with \cref{cor:exp}, prove the theorem.
\end{proof}

We now give sharp tail bounds on the top-$r$ and bottom-$r$ load averages.

\begin{theorem}[Poisson Approximation Tail Bounds]
\label{thm:tail-opt}
     Let $u\ge \ell\ge 1$, $m\ge 1$, $n\coloneqq 2^\ell$, $1\le r\le n$ be integers and denote $\mu \coloneqq m/n$. Let $h: \F_2^u\to\F_2^\ell$ be a uniformly random linear map. For any subset $S\subset\F_2^{u}$ of cardinality $m$ and $\eps\in (0,1)$, we have \[\Pr_h\left[L_r^+(S,h) > \mu + \left(T_r^+\left(\frac{m}\eps, n\right) - \frac{\mu}{\eps}\right)\right] \le \eps\] and \[\Pr_h\left[L_r^-(S,h)  < \mu -\left(\frac{\mu}{\eps}-T_r^-\left(\frac{m}\eps, n\right)\right)\right] \le \eps.\]  
\end{theorem}

\begin{proof}
    For $\star\in \{+,-\}$, let $t^\star = T_r^\star\left(m/\eps, n\right)$. If $r=n$, then $T_r^\star(m/\eps,n) = \mu/\eps$ and $ L_r^\star(S,h) = \mu$, so the theorem follows. Hence, assume $r < n$ from which we have $t^- < \mu/\eps < t^+$. We note that when $\lambda = \log(\eps t^+/\mu) > 0$,

    \[\frac{\frac{\mu}\eps(e^\lambda - 1) + \log(n/r)}{\lambda } = \frac{t^+ - \frac{\mu}\eps + \log(n/r)}{\log(\frac{\eps t^+}\mu)} = t^+,\] where the last equality follows from $D(t^+, \mu/\eps) = \log(n/r)$. Therefore the upper tail follows by setting $\lambda = \log(\eps t^+/\mu) > 0$ in \cref{cor:tail}. 
    
    If $t^- > 0$, plugging $\lambda = \log(\eps t^-/\mu) < 0$ into \cref{cor:tail} and applying the same reasoning yields the lower tail. If $t^-=0$, then
\[
\mu-\left(\frac{\mu}{\varepsilon}-t^-\right)
=\mu-\frac{\mu}{\varepsilon}<0.
\]
Since $L_r^-(S,h)\ge 0$ for every $h$, the lower-tail event
is empty, and the bound follows.
\end{proof}

At this point, the reader interested in getting tight explicit estimates on tail bounds or expectation bounds on the loads in any ball/bin regime can simply disregard linear maps and analyze $T^+_r(m,n)$ or $T^-_r(m,n)$ directly using Taylor approximations of the Lambert $W$ function \cite{LambertW}. For example, in the canonical setting of $m=n$, $T^+(n,n)$ is the real $t\ge 1$ such that $t\log t - t + 1 = \log n$. This solves to $t = \frac{\log n - 1}{W(\frac{\log n - 1}{e})}$. Taylor expanding $W$ and gathering the first few terms of the reciprocal series yields  \[t = \frac{\log n}{\log\log n}\left(1 + \frac{\log\log\log n + 1}{\log\log n} + (1+o(1))\left(\frac{\log\log\log n}{\log\log n}\right)^2\right),\] This sharpens the prior expectation bound of \cite{bshouty26}. Similar calculations can be done for any $m$ and $n$ and to any fixed asymptotic order.  

\subsection{On Finding the Optimal $\lambda$} As promised, we now mention how the expressions in \cref{thm:exp-opt} and \cref{thm:tail-opt} are derived. We will only focus on the upper bounds; the lower bounds are analogous. We rely on the Cram\'{e}r transform, a standard tool in the Chernoff method.

\begin{definition}[Cram\'{e}r transform]
    Let $\psi$ be a convex function. The Cram\'{e}r transform of $\psi$, denoted as $\psi^*$, is defined by \[\psi^*(t) = \sup_{\lambda > 0} (\lambda t - \psi(\lambda)) \]
\end{definition}
We will use the following standard lemma about Cram\'{e}r transforms, provable using calculus.
\begin{lemma}[\hspace{-0.1mm}{\cite[Lemma 2.4]{blm13}}]
\label{lem:cram}
    For a function $f: \R\to\R$, define the generalized inverse \[f^{-1}(y) = \inf \{t\ge 0: f(t) > y\}. \]For a convex and continuously differentiable function $\psi$ such that $\psi(0) = \psi'(0) = 0$ and $y\ge 0$, we have \[(\psi^*)^{-1}(y) = \inf_{\lambda > 0} \frac{\psi(\lambda) + y}\lambda.\] 
\end{lemma}

From this lemma, we can solve the infimum optimization problems in \cref{cor:exp} and \cref{cor:tail}. Letting $P_\mu(\lambda)\coloneqq \mu(e^\lambda -\lambda- 1)$, we see that $P_\mu(0) = P_\mu'(0) = 0$, and thus by \cref{lem:cram},  \[ \inf_{\lambda > 0} \frac{\mu(e^\lambda - 1)+\log(n/r)}\lambda = \mu + \inf_{\lambda > 0} \frac{P_\mu(\lambda)+\log(n/r)}\lambda = \mu + (P_\mu^*)^{-1}(\log(n/r)).\] $P_\mu$ is the centered cumulant generating function of $\pois(\mu)$, and it is known from the literature that for $\nu\ge 0$, 
$P_{\mu}^*(\nu) = D(\mu + \nu,\mu)$ \cite[Section 2.2]{blm13}. Hence, by \cref{def:threshold},
\[
\mu+(P_\mu^*)^{-1}(\log(n/r))=T_r^+(\mu n,n),
\]
which gives the upper bound in \cref{thm:exp-opt}. For the upper bound in \cref{thm:tail-opt}, observe that $P_\mu/\eps=P_{\mu/\eps}$. Applying the same argument to the threshold in \cref{cor:tail} gives
\[
\mu  - \frac{\mu}\eps + \inf_{\lambda>0}\frac{\frac{\mu}\eps(e^\lambda - 1 ) + \log(n/r)}{\lambda} = \mu+\inf_{\lambda>0}
\frac{P_\mu(\lambda)/\eps+\log(n/r)}{\lambda}
=\mu+T_r^+\left(\frac{\mu n}\eps,n\right)-\frac{\mu}{\eps}.
\]

\section{Applications}

In this section, we present some interesting consequences of \cref{thm:exp-opt} and \cref{thm:tail-opt}. The first consequence showcases how tight of expectation bounds we can derive from \cref{thm:exp-opt}. The second consequence gives a strong concentration theorem on bin loads. Optimal two-sided bounds pop out of this lemma as a simple corollary.

\subsection{Additively Optimal Expected Max Load}

We show that in the canonical setting of $n$ balls into $n$ bins, the expected max load of a random linear map is extremely close to the expected max load of a fully random map. To start, we first bound how close $T^+(m,n)$ is to the max load of a fully random map in the sparse regime. As the number of balls decreases with respect to the number of bins, the Chernoff threshold grows with respect to the average load. 
Letting $\mu\coloneqq m/n$ be the average load, and $t \coloneqq T^+(m,n)$ be the Chernoff threshold, we show in the sparse regime where $t/\mu$ grows sufficiently fast, $t$ is within $3/2 + o(1)$ of the expected max load of a fully random map. The proof is inspired by the classic max load lower bounds established via the Poisson Approximation theorem \cite[Lemma 5.12]{mitzupfal}.

\begin{lemma}
\label{lem:Lnn-estimate}
     Let $n, m\coloneqq m(n)$ be integers, $0 < \eta(n)\le 1$, $\mu\coloneqq m/n$ and $t\coloneqq T^+(m,n)$. Define $R(m,n) = \E_f[\max_{i\in [n]} |f^{-1}(i)|]$, where $f: [m]\to[n]$ is a uniformly random function. If $\lim_{n\to\infty}\frac{(t/\mu)^\eta}{\sqrt{\mu} \log(e+t)}\to \infty$, then we have \[R(m,n)\ge  \left\lfloor t- \frac{1}2 - \eta\right\rfloor -o(1).\]
\end{lemma}

    \begin{proof}
     Let $\delta = \frac{1}2 + \eta$ and set $s = \lfloor t-\delta\rfloor$. If $s \le \mu$, then since $R(m,n)\ge \mu$, it follows that $R(m,n)\ge \mu\ge s$. Henceforth, assume $s > \mu$, which also implies $s\ge 1$.
    
    Let $X_1,\dots X_n$ be independent $\pois(\mu)$ random variables. By the Poisson approximation theorem for monotonically decreasing events \cite[Corollary 5.11]{mitzupfal},  \[\Pr\left[\max_{i\in [n]} |f^{-1}(i)| < s\right]\le 2\Pr\left[\max_i X_i < s\right]= 2 (\Pr[X_1 < s])^n\le 2e^{-n\Pr[X_1\ge s]},\]
    where the last inequality used $1-x\le e^{-x}$. We can crudely bound 
    \[R(m,n)  \ge s \cdot \Pr\left[\max_{i\in [n]} |f^{-1}(i)| \ge s\right]    = s - s\Pr\left[\max_{i\in [n]} |f^{-1}(i)| < s\right]  > s - 2s e^{-n\Pr[X_1\ge s]}. \]
   It remains to show $2s e^{-n\Pr[X_1\ge s]} = o(1)$. Using Stirling's formula \cite{robbins}, we estimate
    \[ \Pr[X_1 = s]= e^{-\mu}\frac{\mu^s}{s!}\ge e^{-\mu}\left(\frac{e\mu}{s}\right)^s\frac{1}{e^{1/12s}\sqrt{2\pi s}} = \frac{e^{-D(s, \mu)}}{e^{1/12s}\sqrt{2\pi s}} \ge  \frac{e^{-D(s, \mu)}}{2\sqrt{\pi s}}.\]
    We now use the definition of $t$ to lower bound
\[n\Pr[X_1 \ge s]  \ge e^{D(t,\mu)} \cdot \Pr[X_1 = s] \ge \frac{e^{D(t,\mu)-D(s, \mu)}}{2\sqrt{\pi s}}\ge \frac{e^{\delta \log(s/\mu)}}{2\sqrt{\pi s}} =  \frac{1}{2\sqrt{\pi \mu}} \left(\frac{s}\mu\right)^\eta,\] where the last inequality follows from \cref{prop:kl-est}. As $s\ge 1$ and $s = \lfloor t-\delta\rfloor$, it follows $s\ge t/4$. Hence, by the limit assumption, \[n \Pr[X_1\ge s]\ge \frac{1}{2\sqrt{\pi \mu}} \left(\frac{s}\mu\right)^\eta\ge\frac{1}{2\cdot 4^\eta\sqrt{\pi \mu}} \left(\frac{t}\mu\right)^\eta = \omega(\log(e+t)). \] Therefore, 

\[R(m,n)\ge s- 2se^{-n\Pr[X_1\ge s]} \ge s -  2t e^{-\omega(\log(e+t))} = s - o(1).\]

\end{proof}

We can now deduce \cref{thm:headline}.

\begin{proof}[Proof of \cref{thm:headline}]
Denote $t = T^+_1(n,n)$. Let $\eta(n) = \min\left\{1,\sqrt{1/\log t}\right\} = o(1)$. We see $\frac{t^\eta}{\log(e+t)} = \frac{e^{\sqrt{\log t}}}{\log(e+t)}\to \infty$ as long as $t\to \infty$. Indeed, $D(t,1) = \log n$ implies \[\log n = t\log (t/e) + 1 <  t^2/e+1\implies t > \sqrt{e\log (n/e)} = \omega(1), \]
    Hence, we can invoke \cref{lem:Lnn-estimate} and \cref{thm:exp-opt} with $m=n$ to get that \[\E[L_1^+(S,h)]\le t\le \E[L_1^+([n], f)] + t -  \left\lfloor t- \frac{1}2 - \eta\right\rfloor +o(1) =  \E[L_1^+([n], f)] +\frac{1}2 + c_\ell +o(1),\] where $c_\ell = \left\{t-\frac{1}2-\eta\right\} = \left\{T^+_1(2^\ell,2^\ell)-\frac{1}2-\eta(2^\ell)\right\}$. Denote $t_\ell\coloneqq T^+_1(2^\ell,2^\ell)$, $\eta_\ell \coloneqq \eta(2^\ell)$, and $d_\ell \coloneqq t_\ell - \frac{1}2 - \eta_\ell$, whence $c_\ell = \{d_\ell\}$. Notice that to show a subsequence of $(c_\ell)$ goes to $0$, it suffices to show that $(d_\ell)$ is monotone increasing, $d_\ell\to \infty$, and $d_{\ell + 1}- d_\ell\to 0$. With these three claims, for each integer we can take the first $\ell$ at which $d_\ell$ crosses that integer; since the increments are $o(1)$, the overshoot (and hence the corresponding fractional parts) approach 0. The first two claims are true as $t_\ell$ monotonically increases to $\infty$ and $\eta_\ell$ monotonically decreases to $0$. For the final claim, we first observe \[\log 2 = \log(2^{\ell + 1}) - \log(2^\ell) = D(t_{\ell + 1},1) - D(t_\ell, 1) \ge (t_{\ell + 1} - t_\ell)\log t_\ell\] by \cref{prop:kl-est}. Therefore, \[d_{\ell + 1} - d_\ell = t_{\ell + 1} - t_\ell - \eta_{\ell + 1} + \eta_\ell \le \frac{\log 2}{\log t_\ell} + o(1)  = o(1)\] since $t_\ell\to\infty$. With the three claims proven, we now conclude.
\end{proof}

\begin{remark}
\label{rem:sparse-regime}
We note that the same $\frac{1}2 + c_\ell + o(1)$ gap can be derived for small enough $m$ relative to $n$ (as dictated by the limit condition in \cref{lem:Lnn-estimate}). We do not include the explicit calculations here.
\end{remark}
\subsection{Concentration}
In this section, we use \cref{thm:exp-opt} to prove sharp load concentration and uniformity in the dense regime when $m\ge n\log n$. First, we will need some estimates on $T^+_r(m,n)$ and $T^-_r(m,n)$ in the dense regime. 

\begin{proposition}
\label{prop:approx}
    Let $m$ be a positive real, $n \ge r\ge 1$ be integers, and $\mu \coloneqq m/n$. When $\mu\ge \log(n/r)$, we have \[T^+_r(m,n) \le  \mu + 2\sqrt{\mu \log (n/r)}\] and \[T^-_r(m,n)\ge \mu - \sqrt{2\mu \log(n/r)}.\]
\end{proposition}

\begin{proof}
Let $a \coloneqq T^+_r(m,n) - \mu$, $b \coloneqq \mu - T^-_r(m,n)$. We first bound $b$ since it is easier. By using \cref{prop:kl-est}, we estimate \begin{align*}\log (n/r)  = D(\mu-b,\mu) - D(\mu,\mu)  &  = \int_{\mu}^{\mu-b} \log(x/\mu)dx \\ & =  -\mu \int_0^{b/\mu} \log(1-u)du  \\ & \ge \mu \int_0^{b/\mu}udu  = \frac{b^2}{2\mu}. \end{align*} Thus $b \le \sqrt{2\mu \log (n/r)}$. We now use a similar calculation to bound $a$. Set $a_0\coloneqq 2\sqrt{\mu\log (n/r)}$. Since $\mu\ge\log (n/r)$, we have $a_0/\mu\le 2$. By noting $\log(1+u)\ge u/2$ for $u\in[0,2]$, and using \cref{prop:kl-est}, we see
\[D(\mu+a_0,\mu)=\mu\int_0^{a_0/\mu}\log(1+u)\,du\ge \frac{\mu}{4}\left(\frac{a_0}{\mu}\right)^2=\log (n/r).\]
Because $a\mapsto D(\mu+a,\mu)$ is increasing and $D(\mu+a,\mu)=\log(n/r)$, it follows that $a\le a_0=2\sqrt{\mu\log (n/r)}$.

\end{proof}

\begin{remark}
    By combining \cref{prop:approx} with \cref{thm:exp-opt}, we obtain that the expected max load for $m = \omega(n\log n)$ is bounded by $\frac{m}n + 2\sqrt{\frac{m\log n}n} = (1+o(1))\frac{m}n$, strongly resolving a question of Pathegama and Barg over $\F_2$ (see Remark 6 and the preceding paragraph of \cite{PaBa25}).
\end{remark}

We now prove our concentration result: the number of balls in any density-$\rho$ subset of bins will not deviate far from expectation.

\begin{theorem}
\label{thm:mixing}
    Let $u\ge \ell \ge 1, n\coloneqq 2^\ell$ be integers, $\eps\in (0,1), \rho\in (0,1)$ such that $\rho n\in \Z$, and $A\subset \F_2^u$ be an arbitrary subset of size $m$. Assume $m > n\log(1/\rho)$. For a uniform random linear $h: \F_2^u\to\F_2^\ell$, we have \[\Pr_h\left[\forall B\subset\F_2^\ell\text{ s.t. } |B|=\rho n, \left|\Pr_{a\sim A}[h(a)\in B] - \rho\right| \le 5\rho \sqrt{\frac{n \log (1/\rho)}{\eps m}}\right]\ge 1 - \eps.\]
\end{theorem}

\begin{proof}
    Denote $r\coloneqq \rho n$ and $\mu\coloneqq m/n$. Notice that $\Pr_{a\sim A}[h(a)\in B] = \frac{|A\cap h^{-1}(B)|}{|A|}$ is maximized when $B$ consists of the $r$ heaviest bins. In this case, we will have \[\Pr_{a\sim A}[h(a)\in B] = \frac{|A\cap h^{-1}(B)|}{m} = \frac{rL_{r}^+(A,h)}{m}.\] Similarly, the minimum value of the expression is obtained when $B$ consists of the $r$ lightest bins, whence \[\Pr_{a\sim A}[h(a)\in B] = \frac{rL_{r}^-(A,h)}{m}.\] Therefore, for all $B\subset \F_2^\ell$ of size $r$,  \begin{align*}\left|\Pr_{a\sim A}[h(a)\in B] - \rho\right| &\le \max\left\{\frac{rL_{r}^+(A,h)}m - \frac{r}n, \frac{r}n - \frac{rL_{r}^-(A,h)}m\right\}  \\ & = \frac{r}m\max\left\{L_{r}^+(A,h) - \mu, \mu - L_{r}^-(A,h)\right\} \\ & \le \frac{r}m(L^+_{r}(A,h) - L^-_{r}(A,h)). \end{align*} We know from \cref{thm:tail-opt} and a union bound is that for each $r\in [n]$, \[\Pr\left[ L^+_r(A,h) - L^-_r(A,h) > T_r^+(2m/\eps, n) - T_r^-(2m/\eps, n)\right]\le \frac{\eps}2 + \frac{\eps}2 = \eps.\]   Therefore, with probability at least $1-\eps$, we have simultaneously for all $B\subset \F_2^\ell$ of density $\rho$ that \begin{align*}\left|\Pr_{a\sim A}[h(a)\in B] - \rho\right| & \le \frac{r}m (T_{r}^+(2m/\eps, n) - T_{r}^-(2m/\eps, n)) \\ & \le \frac{5r}m\sqrt{\frac{\mu\log(n/r)}\eps}  \\ & = 5\rho\sqrt{\frac{n\log(1/\rho)}{\eps m}}, \end{align*} where we used \cref{prop:approx} in the second inequality.
\end{proof}

A corollary of this result is tight two-sided bounds.

\begin{corollary}
\label{thm:twosided}
 Let $u\ge \ell\ge 1$, $n\coloneqq 2^\ell$, $m > n\log n$ be integers, and $A\subset\F_2^{u}$ be an arbitrary set of cardinality $m$. For every $\eps\in(0,1)$, a uniformly random linear map $h: \F_2^u\to\F_2^\ell$ satisfies \[\Pr\left[\forall y\in \F_2^\ell, \left|L_y(A,h) - \frac{m}n\right| \le 5\sqrt{\frac{m\log n}{\eps n}}\right]\ge 1-\eps .\]
\end{corollary}

\begin{proof}
    The corollary follows by considering all singleton sets $B = \{b\}$ in \cref{thm:mixing}. With probability $\ge 1-\eps$ we have \[ \left|\Pr_{a\in A}[h(a) = b] - \frac{1}n\right|\le 5\sqrt{\frac{\log n}{\eps mn}}\] for all $b\in \F_2^\ell$. Multiplying both sides by $m$ gives \[ \left|L_b(A,h) - \frac{m}n\right| \le 5\sqrt\frac{ m\log n}{\eps n}\] for all $b\in \F_2^\ell$, giving the desired result.

\end{proof}

By a Chernoff bound followed by a union bound over all sets of density $\rho$, it is not hard to see that for constant $\eps$, \cref{thm:mixing} matches the performance of a fully random function. A formal lower bound can be obtained by considering a random set of size $m$ and using the probabilistic method. However, the deviation of a random map has a much better dependence on $\eps$, so it is natural to ask if the dependence on $\eps^{-1/2}$ is tight. We show that this is the case up to a log factor.

\begin{theorem}
\label{thm:mixing-tight}
Let $m\ge n\coloneqq 2^\ell$ be powers of two,
$\eps\in(\frac{n}{88m},\frac{1}{44})$, $\rho\in (0,1/2]\cap \frac{1}n \Z$. For every sufficiently large $u$ (in terms of
$m,n,\eps,\rho$), there exists $A\subset \F_2^u$ of cardinality
$m$ such that for uniformly random linear
$h: \F_2^u\to\F_2^\ell$, we have
\[
\Pr_h\left[
\exists B\subset\F_2^\ell\text{ s.t. } |B|=\rho n,\ 
\left|\Pr_{a\sim A}[h(a)\in B]-\rho\right|
\ge \frac{\rho}{4}\sqrt{\frac{n}{33\eps m}}
\right]\ge \eps.
\]
\end{theorem}

\begin{proof}
   We will take $A$ to be $k$ linearly independent translates of one vector space of size $m/k$, where $1\le k\le m/n$ is a power of two to be chosen later. More formally, let $v\coloneqq \log_2 (m/k)$ and $V = \F_2^{v}$. Since $m\ge kn$, we have $v\ge \ell$. For sufficiently large $u$, we can write $\F_2^u = \F_2^{u-v}\oplus V$ and take $e_1,\dots,e_k$ to be distinct standard basis vectors of $\F_2^{u-v}$. Let $A = \bigcup_{i=1}^k (V+e_i)$. These translates are disjoint, so $|A|=m$. We will show that with probability at least $\frac{nk}{4m}$, $h(V)$ has dimension $\ell - 1$. Conditioned on this, we show with constant probability that $\left|\Pr_{a\sim A}[h(a)\in B] - \rho\right| \ge \Omega\left(\frac{\rho}{\sqrt{k}}\right)$. Setting $k = \Theta(\eps m/n)$ then yields the desired result.

   The probability $h(V)$ has dimension $\ell-1$ is equal to the probability a random element of $\F_2^{\ell \times v}$ has rank $\ell - 1$. This is lower bounded by the probability the first $\ell-1$ rows are linearly independent, and the last row is in the span of the previous ones. By iteratively looking at the probability the next vector is good given previous ones, we get the lower bound \[(1-2^{-v})(1-2^{-v + 1})\dots (1-2^{-v + \ell-2}) 2^{-v + \ell - 1} \ge 2^{-v + \ell - 1}\left(1 - \sum_{i=0}^{\ell-2} 2^{-v + i}\right)\ge 2^{-v + \ell-2}.\] 

   Conditioning on this event, we denote the hyperplane $H_0 \coloneqq h(V)$. Let $H_1$ be the one other coset of $H_0$ (so that $\F_2^\ell = H_0 \sqcup H_1$). For each $i$, $h(V+e_i) = H_0 + h(e_i) = H_{c_i}$ for some $c_i\in \{0,1\}$, depending on whether $h(e_i)\in H_0$ or not. Furthermore, $h$ will hash $V+e_i$ uniformly into $H_{c_i}$. This is seen by noting $h|_V$ is a linear map mapping $V$ uniformly into $H_0$. Since $h(V+e_i) = h(V) + h(e_i)$, and affine shifts do not change the load distribution of a linear map, it follows $h$ hashes $V+e_i$ evenly into $H_{c_i}$. 
   
   Picking uniform $a\sim A$ is equivalent to picking uniform $i\sim [k]$, and then picking uniform $a\sim (V+e_i)$. Hence, for fixed $h$, picking $h(a)$ for $a\sim A$ can be thought of as picking uniform $i\sim [k]$, and then picking uniform $y\sim H_{c_i}$. Equivalently, this is the same as randomly picking $c\in \{0,1\}$ where $c=1$ with probability $\frac{1}k\sum_{i=1}^k c_i$, and then picking uniformly from $H_c$. This latter perspective is what we will adopt.

   Since $e_1,\dots e_k$ are linearly independent of $V$, $h(e_1),\dots, h(e_k)$ are uniform and independent vectors in $\F_2^\ell$, even when conditioned on $h(V) = H_0$. Hence $\{c_i\}_{i\in [k]}$ are independent coin flips. By the fourth-moment method \cite[Corollary 3.7]{Berger97} and symmetry, we have $\frac{1}k\sum_{i\in [k]}c_i - \frac{1}2 \ge \frac{1}{2\sqrt{3k}}$ with probability at least $1/11$. Condition on this event, which intuitively forces the random variable $c$ to be biased towards 1.

   Let $B$ be an arbitrary subset of $H_1$ of size $\rho n \le n/2$. The probability $h(a)\in B$ for $a\sim A$ is equivalent to the probability $c=1$, multiplied by the probability a uniform element of $H_1$ is in $B$. The former probability is at least $\frac{1}2 + \frac{1}{2\sqrt{3k}}$ and the latter probability is $\frac{\rho n}{n/2} = 2\rho$. Therefore, \[|\Pr_{a\sim A}[h(a)\in B]-\rho|\ge \left(\frac{1}2 + \frac{1}{2\sqrt{3k}}\right) 2\rho - \rho = \frac{\rho}{\sqrt{3k}}. \] 

   We now recap the events we conditioned on. With probability $\ge 2^{-v+ \ell-2} = \frac{nk}{4m}$, $h(V)$ is of codimension 1. Conditioning on this, with probability at least $1/11$ we have $\frac{1}k\sum_{i\in [k]}c_i\ge \frac{1}2+\frac{1}{2\sqrt{3k}}$, from which it follows that there exists $B$ with density $\rho$ such that $|\Pr_{a\sim A}[h(a)\in B]-\rho|\ge \frac{\rho}{\sqrt{3k}}$. Hence, we have shown that for our constructed $A$, \[\Pr_h\left[\exists B\subset\F_2^\ell\text{ s.t. } |B|=\rho n, \left|\Pr_{a\sim A}[h(a)\in B] - \rho\right| \ge \frac{\rho}{\sqrt{3k}}\right]\ge  \frac{nk}{44m}.\] Set $k$ to be the smallest power of two not less than $\frac{44\eps m}n$, making $\frac{44\eps m}n\le k< \frac{88\eps m}n$. Thus $\frac{nk}{44m}\ge\eps$ and
   \[\frac{\rho}{\sqrt{3k}}\ge\rho\sqrt{\frac{n}{264\eps m}}\ge\frac{\rho}{4}\sqrt{\frac{n}{33\eps m}}.\] Since $\eps>\frac{n}{88m}$, $44\eps<1$, and $\frac{m}n$ is a power of two, we also have $1\le k\le\frac{m}n$,
   giving the desired result.
\end{proof}

\section*{Acknowledgements}
We thank Nader H. Bshouty, who observed that our initial bound of $3/2+o(1)$ in \cref{thm:headline} could be improved to the oscillatory bound presented. We thank David Zuckerman for discussions and directing us to the textbook of Boucheron-Lugosi-Massart \cite{blm13}. We thank Venkatesan Guruswami for asking about more general properties of bins that potentials can handle, which spurred us to consider top-$r$ averages. We also thank Dean Doron and Jesse Goodman for various discussions. We thank Geoffrey Mon and Amit Rajaraman for advice regarding the presentation of this work.

\subsection*{AI Methodology}

All mathematical ideas and writing are due to the author. Following the completion of an initial draft, the author used ChatGPT 6 to verify proofs, copyedit prose, and help interpret the results in Lund's manuscript~\cite{lund}. 

\bibliographystyle{alphaurl} \bibliography{references}

\end{document}